\documentclass[a4paper]{article}
 \usepackage{fullpage}
\usepackage[utf8]{inputenc}
\usepackage{graphicx}
\usepackage{amsmath}
\usepackage{comment}
\usepackage{amssymb}
\usepackage{amsthm}
\usepackage{authblk}
\usepackage{tikz}
\usetikzlibrary{shapes.geometric}

\usepackage[noend, ruled]{algorithm2e}
\DontPrintSemicolon

\newtheorem{theorem}{Theorem}
\newtheorem{lemma}[theorem]{Lemma}
\newtheorem{fact}[theorem]{Fact}

\theoremstyle{remark}
\newtheorem{remark}[theorem]{Remark}
\newtheorem{example}[theorem]{Example}

\newcommand{\lcp}{\mathsf{lcp}}
\newcommand{\Oh}{\mathcal{O}}
\newcommand{\dd}{\mathinner{.\,.}}
\newcommand{\nat}{\mathbb{N}}
\newcommand{\integ}{\mathbb{Z}}
\newcommand{\dist}{\mbox{\sf ED}}
\newcommand{\ed}{\mbox{\sf ed-dist}}

\newcommand{\defproblem}[3]{
  \par\medskip
  \noindent
  \fbox{
    \begin{minipage}{0.96\textwidth}
    {#1}

    \smallskip
    \noindent
    \textbf{Input:} #2

    \smallskip
    \noindent
    \textbf{Output:} #3
    \end{minipage}
  }
  \smallskip
}

\usepackage[hidelinks]{hyperref}
\usepackage[capitalise]{cleveref}

\date{\vspace{-1cm}}

\begin{document}
\title{
 Faster Recovery of Approximate Periods over Edit Distance}
%
%

  \author{Tomasz Kociumaka}
  \author{Jakub Radoszewski\thanks{Supported by the ``Algorithms for text processing with errors and uncertainties''
 project carried out within the HOMING programme of the Foundation for Polish Science co-financed by the European Union under the European Regional Development Fund.}}
  \author{Wojciech Rytter\thanks{Supported by the Polish National Science Center, grant no.\ 2014/13/B/ST6/00770.}}
  \author{Juliusz Straszyński\protect\footnotemark[1]}
  \author{Tomasz~Waleń}
  \author{Wiktor Zuba}
%
%
  \affil{\normalsize Faculty~of Mathematics, Informatics and Mechanics,
    University of Warsaw, Warsaw, Poland\\
    \texttt{[kociumaka,jrad,rytter,jks,walen,w.zuba]@mimuw.edu.pl}}
%
%
\maketitle              

\begin{abstract}
The approximate period recovery problem asks to compute all \emph{approximate word-periods} of a given word $S$ of length $n$:
all primitive words $P$ ($|P|=p$) which have a periodic extension at edit
distance smaller than $\tau_p$ from $S$, where $\tau_p = \lfloor \frac{n}{(3.75+\epsilon)\cdot p} \rfloor$ for some $\epsilon>0$.
Here, the set of periodic extensions of $P$ consists of all finite prefixes of~$P^\infty$.

We improve the time complexity of the fastest known algorithm for this problem of Amir et al. [Theor. Comput. Sci., 2018] from $\Oh(n^{4/3})$ to $\Oh(n \log n)$. 
Our tool is a fast algorithm for Approximate Pattern Matching in Periodic Text. 
We consider only verification for the period recovery problem when the candidate approximate word-period $P$ is explicitly given up to cyclic rotation;
the algorithm of Amir et al. reduces the general problem in $\Oh(n)$ time to a logarithmic number of such more specific instances.

\end{abstract}

\section{Introduction}
The aim of this work is computing periods of words in the approximate pattern matching model (see e.g.\ \cite{Jewels,DBLP:books/cu/Gusfield1997}).
This task can be stated as the \emph{approximate period recovery (APR) problem} that was defined by Amir et al.~\cite{DBLP:journals/talg/AmirELPS12}.
In this problem, we are given a word; we suspect that it was initially periodic, but then errors might have been introduced in it.
Our goal is to attempt to recover the periodicity of the original word.
If too many errors have been introduced, it might be impossible to recover the period.
Hence, a requirement is imposed that the distance between the original periodic word and the word with errors is upper bounded, with the bound being related to the period length.
Here, edit distance is used as a metric.
The fastest known solution to the APR problem is due to Amir et al.~\cite{DBLP:journals/tcs/AmirALS18}.

A different version of the APR problem was considered by Sim et al.~\cite{DBLP:journals/tcs/SimIPS01}, who bound the number of errors per occurrence of the period.
The general problem of computing approximate periods over weighted edit distance is known to be NP-complete;
see~\cite{DBLP:journals/tcs/Popov03,DBLP:journals/tcs/SimIPS01}.
Other variants of approximate periods have also been introduced.
One direction is the study of approximate repetitions, that is, subwords of the given word that are
approximately periodic in some sense (and, possibly, maximal); see~\cite{DBLP:journals/tcs/AmitCLS17,DBLP:journals/tcs/KolpakovK03,DBLP:journals/bioinformatics/SokolBT07,DBLP:journals/tcs/SokolT14}.
Another is the study of quasiperiods, occurrences of which may overlap in the text;
see, e.g.,~\cite{DBLP:journals/tcs/ApostolicoE93,DBLP:journals/ipl/ApostolicoFI91,DBLP:journals/ipl/Breslauer92,DBLP:conf/soda/KociumakaKRRW12,DBLP:journals/algorithmica/LiS02}.

Let $\ed(S,W)$ be the edit distance (or Levenshtein distance) between the words $S$ and $W$, that is, the minimum number of edit operations
(insertions, deletions, or substitutions) necessary to transform $S$ to $W$.
A word $P$ is called \emph{primitive} if it cannot be expressed as $P=Q^k$ for a word $Q$ and an integer $k \ge 2$.
The APR problem can now formally be defined as follows.

\defproblem{Approximate Period Recovery (APR) Problem}{
  A word $S$ of length $n$
}{
  All primitive words $P$ (called \emph{approximate word-periods}) for which the infinite word $P^\infty$ has a prefix $W$ such that $\ed(S,W) < \tau_p$, where
  $p=|P|$ and $\tau_p = \lfloor \frac{n}{(3.75+\epsilon)\cdot p} \rfloor$ with $\epsilon > 0$
}

\begin{remark}
  Amir et al.~\cite{DBLP:journals/tcs/AmirALS18} show that each approximate word-period is a subword of $S$ and thus can be represented in constant space.
  Moreover, they show that the number of approximate word-periods is $\Oh(n)$.
  Hence, the output to the APR problem uses $\Oh(n)$ space.
\end{remark}
The solution of Amir et al.~\cite{DBLP:journals/tcs/AmirALS18} works in $\Oh(n^{4/3})$ time\footnote{Also the APR problem under the Hamming distance
was considered~\cite{DBLP:journals/talg/AmirELPS12} for which an $\Oh(n \log n)$-time algorithm was presented~\cite{DBLP:journals/tcs/AmirALS18}
for the threshold $\lfloor \frac{n}{(2+\epsilon)\cdot p} \rfloor$ with $\epsilon>0$.}.
Our result is an $\Oh(n \log n)$-time algorithm for the APR problem.

Let us recall that two words $U$ and $V$ are \emph{cyclic shifts} (denoted as $U \approx V$) if there exist words $X$ and $Y$ such that $U=XY$ and $V=YX$.
The algorithm of Amir et al.~\cite{DBLP:journals/tcs/AmirALS18} consists of two steps.
First, a small number of candidates are identified, as stated in the following fact.

\begin{fact}[Amir et al.~{\cite[Section 4.3]{DBLP:journals/tcs/AmirALS18}}]\label{fct:red}
  In $\Oh(n)$ time, one can find $\Oh(\log n)$ subwords of~$S$ (of exponentially increasing lengths) such that every approximate word-period of~$S$
  is a cyclic shift of one of the candidates.
\end{fact}

For a pattern $S$ and an infinite word $W$, by $\dist(S,W)$ let us denote the minimum edit distance between $S$ and a prefix of $W$.
By \cref{fct:red}, the APR problem reduces to $\Oh(\log n)$ instances of the following problem.

\defproblem{Approximate Pattern Matching in Periodic Text (APM Problem)}{
  A word $S$ of length $n$, a word $P$ of length $p$, and a threshold $k$
}{
  For every cyclic shift $U$ of $P$, compute $\dist(S,U^\infty)$ or report that this value is greater than $k$
}

Amir et al.~\cite{DBLP:journals/tcs/AmirALS18} use two solutions to the APM problem that work in $\Oh(np)$ time and $\Oh(n+k(k+p))$ time,
respectively. 
The main tool of the first algorithm is \emph{wrap-around dynamic programming}~\cite{DBLP:journals/ipl/FischettiLSS93}
that solves the APM problem without the threshold constraint $k$ in $\Oh(np)$ time.
The other solution is based on the Landau--Vishkin algorithm~\cite{DBLP:journals/jal/LandauV89}.
For each $p$ and $k < \tau_p$, either algorithm works in $\Oh(n^{4/3})$ time.

\paragraph{Our results.}
We show that:
\begin{itemize}
  \item The APM problem can be solved in $\Oh(n+kp)$ time.
  \item The APR problem can be solved in $\Oh(n \log n)$ time.
\end{itemize}
Our solution to the APM problem involves a more efficient combination of wrap-around dynamic programming with the Landau--Vishkin algorithm.

\section{Approximate Pattern Matching in Periodic Texts}

We assume that the length of a word $U$ is denoted by $|U|$ and the letters of $U$ are numbered $0$ through $|U|-1$,
with $U[i]$ representing the $i$th letter.
By $U[i \dd j]$ we denote the subword $U[i] \cdots U[j]$; if $i>j$, it denotes the empty word.
A prefix of $U$ is a subword $U[0 \dd i]$ and a suffix of $U$ is a subword $U[i \dd |U|-1]$, denoted also as $U[i \dd]$.

The length of the longest common prefix of words $U$ and $V$ is denoted by $\lcp(U,V)$.
The following fact specifies a well-known efficient data structure answering such queries over suffixes of a given text;
see, e.g., \cite{AlgorithmsOnStrings}.

\begin{fact}\label{fct:ver}
Let $S$ be a word of length $n$ over an integer alphabet of size $\sigma = n^{\Oh(1)}$.
After $\Oh(n)$-time preprocessing, given indices $i$ and $j$ ($0 \le i,j < n$) one can compute $\lcp(S[i \dd],S[j \dd])$ in $\Oh(1)$ time.
\end{fact}

\subsection{Wrap-Around Dynamic Programming}\label{sec:wrap}

Following~\cite{DBLP:journals/ipl/FischettiLSS93}, we introduce a table $T[0\dd n, 0\dd p-1]$
whose cell $T[i,j]$  denotes the minimum edit distance between $S[0\dd i-1]$ and some 
subword of the periodic word $P^\infty$ ending on the $(j-1)$th character of the period.
More formally, for $i \in \{0,\ldots,n\}$ and $j \in \integ_p$, we define
\[T[i,j]=\min \{ \ed(S[0\dd i-1],P^\infty[i'\dd j'])\,:\,i' \in \nat,\ j'\equiv j-1\ (\bmod\,p)\};\]
see \cref{fig:sroda}. The following fact characterizes $T$ in terms of $\dist$.

\begin{fact}\label{fct:obs}
We have $\min\{\dist(S,U^\infty) : U \approx P\} = \min\{T[n,j] : j \in \integ_p\}$.
\end{fact}
\begin{proof}
First, let us observe that the definition of $T$ immediately yields
\[\min\{T[n,j] : j \in \mathbb{Z}_p\}= \min\{\ed(S, P^\infty[i'\dd j'])\,:\, i',j'\in \nat\}.\]
In other words, $\min\{T[n,j] : j \in \mathbb{Z}_p\}$ is the minimum edit distance
between $S$ and any subword of $P^\infty$. 
On the other hand, $\min\{\dist(S,U^\infty) : U \approx P\}$ by definition of $\dist$
is the minimum edit distance between $S$ and a prefix of $U^\infty$ for a cyclic shift $U$ of $P$.
Finally, it suffices to note that the sets of subwords of $P^\infty$ and of prefixes of $U^\infty$ taken over all $U \approx P$ are the same.
%
\end{proof}

Below, we use $\oplus$ and $\ominus$ to denote operations in $\mathbb{Z}_p$.

\begin{lemma}[\cite{DBLP:journals/ipl/FischettiLSS93}]\label{lem:Fis}
The table $T$ is the unique table satisfying the following formula:
\begin{align*}
T[0,j]& =0,\\
T[i+1,j\oplus 1]&=\min\left\{
\begin{matrix}
T[i,j\oplus 1]&+&1 \\
T[i,j]&+& [S[i] \neq P[j]] \\
T[i+1,j]&+&1 
\end{matrix}\right\}.
\end{align*}
\end{lemma}

\begin{figure}[t]
\centering
\includegraphics[width=.45\textwidth]{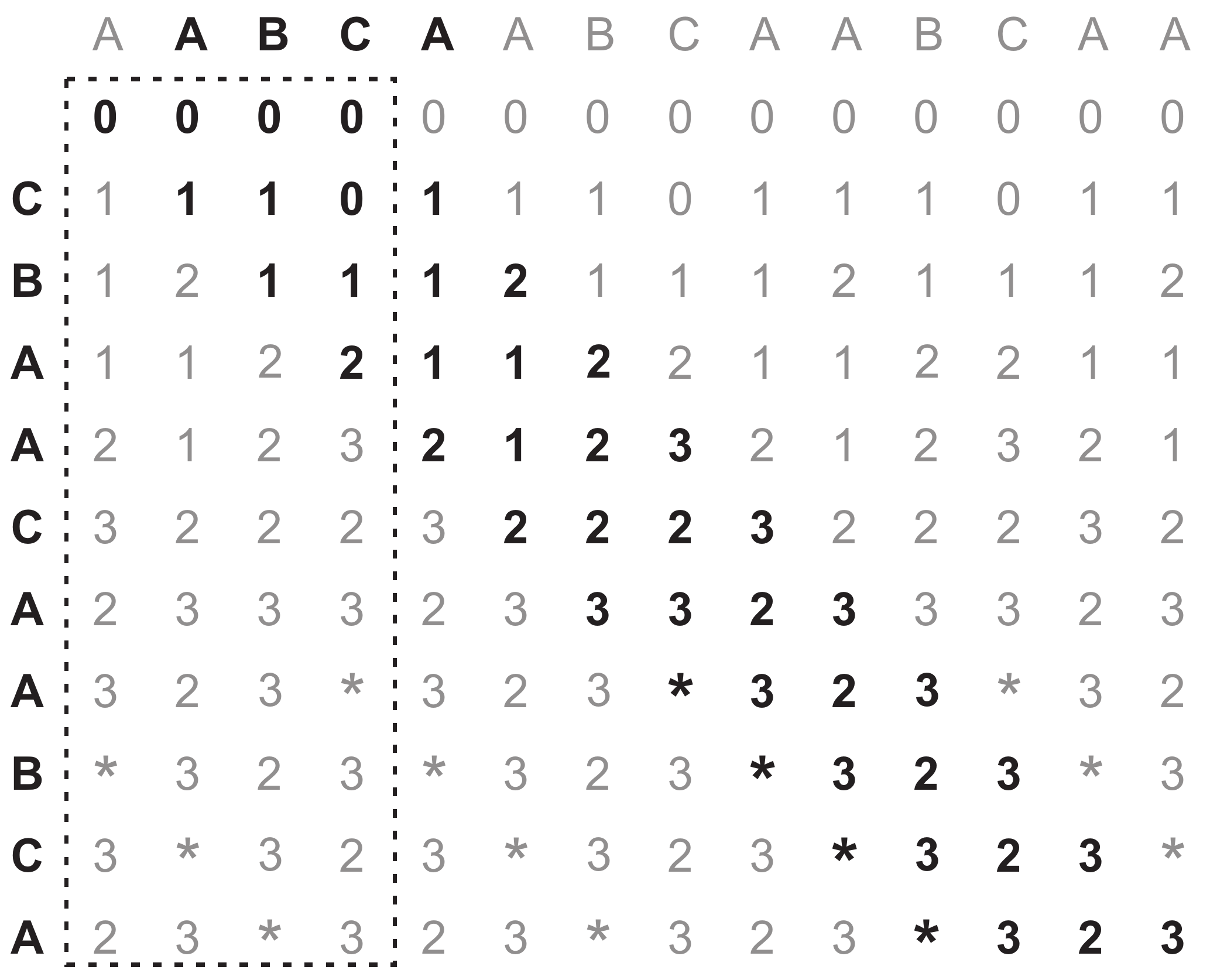}
\caption{The first four columns show the table $T$ for $S=\mathtt{CBAACAABCA}$ and $P=\mathtt{ABCA}$. The asterisks represent values that are greater than $k=3$;
these values need not be computed in our algorithm. The next columns contain copies of $T$; the highlighted diagonals show the computation of the array $D$ (see below).
  Note that $T[3,1]=1$ because $T[2,0] = 1$ and $S[2]=\texttt{A}=P[0]$.
  }\label{fig:sroda}
\end{figure}

Let us mention that the above formula contains cyclic dependencies that emerge due to wrapping (the third value in the minimum).
Nevertheless, the table can be computed using a graph-theoretic interpretation.
With each $T[i,j]$ we associate a vertex $(i,j)$.
The arcs are implied by the formula in \cref{lem:Fis}:
the arcs pointing to $(i+1,j\oplus 1)$ are from $(i,j\oplus 1)$ with weight 1 (deletion), from $(i,j)$ with weight 0 or 1 (match or substitution), and from $(i+1,j)$ with weight 1 (insertion).
Then $T[i, j]$ is the length of the shortest path from any vertex $(0,j')$ to the vertex $(i,j)$.
With this interpretation, the table $T$ is computed using Breadth-First Search, with the 0-arcs
processed before the 1-arcs.

\subsection{Wrap-Around DP with Kangaroo Jumps}
Our next goal is to compute all the values $T[n,j]$ not exceeding $k$.
In the algorithm, we exploit two properties that our dynamic
programming array has. 

First of all, let us consider a diagonal modulo length of the period,
that is, cells of the form $T[i, j \oplus i]$ for a fixed $j\in \integ_p$. We can notice that
the sequence of values on every diagonal is non-decreasing. This stems from the fact that on each diagonal the alignment of 
the pattern is the same and extending a prefix of $S$ and a subword of $P^\infty$ by one letter does not decrease their edit distance.
This results in a conclusion that if we would like to
iteratively compute the set of reachable cells within a fixed distance, then we can
convey this information with just the indices of the furthermost reachable
cells in each of the diagonals.
Our task is to check whether we can reach some cell in the last row
within the distance $k$. To achieve this, we can iteratively find the set of cells reachable
within subsequent distances $0,1,\ldots$. 
More formally, for $d \in \{0,\ldots,k\}$ and $j\in \mathbb{Z}_p$, we define
\[D[d,j]=\max\{i :  T[i,j\oplus i]\le d\};\]
see \cref{fig:sroda}.

Secondly, we observe that it is cheap to check 
how many consecutive cost-0 transitions can be made from a given cell.
Let us remind ourselves that our only free transition checks whether
the next letters of the pattern and the periodic word are equal. To know how far
we can stack this transition is, in other words, finding the longest common prefix of
appropriate suffixes of $S$ and $P^\infty$.
We obtain the following recursive formulae for $D[d,j]$; see Fig.~\ref{fig:d}.

\begin{figure}[b!]
\centering
\begin{tikzpicture}[scale=0.4,square/.style={minimum width=1.05cm, inner sep=-50, regular polygon,regular polygon sides=4}]
  \begin{scope}[xshift=-2cm]
    \draw (0,1) node[above] {$j \ominus 1$};
    \draw (0,0) node[draw, square](l0) {0};
    \draw (10,-10) node[draw,square](ld) {$d$};
    \draw (l0) -- (ld);
   
    \draw[densely dotted] (-3.2,-10) -- (ld);
    \draw (-8,-10) node[right] {$D[d,j \ominus 1]$};
    \draw (12,-12) node[draw, square](ldd) {$d\!+\!1$};
    \draw[thick,-latex] (ld) -- (12,-10);
  \end{scope}
  \begin{scope}[xshift=0cm]
    \draw (0,0) node[draw, square](m0) {0};
    \draw (0,1) node[above] {$j$};
    \draw (4,-4) node[draw, square](md) {$d$};
    \draw (m0) -- (md);
    \draw[densely dotted] (-5.2,-4) -- (3,-4);
    \draw (-5.2,-4) node[left] {$D[d,j]$};
    \draw (6,-6) node[draw, square](mdd) {$d\!+\!1$};
    \draw[thick,-latex] (mdd) -- (8,-8);
    \draw[densely dashed] (8,-8) -- (14,-14);

    \draw (18,-18) node[draw, square](mddd) {$d\!+\!1$};
    \draw[thick,-latex] (14.1,-14.1) -- (mddd);
    \draw[densely dotted] (-5.2,-18) -- (mddd);
    \draw (-5.2,-18) node[left] {$D[d+1,j]$};
      \draw[densely dotted] (-5.2,-14) -- (14,-14);
    \draw (-5.2,-14) node[left] {$i$};
  \end{scope}
  \begin{scope}[xshift=2cm]
    \draw (0,0) node[draw, square](r0) {0};
    \draw (0,1) node[above] {$j \oplus 1$};
    \draw (12,-12) node[draw, square](rd) {$d$};
    \draw (r0) -- (rd);
   
    \draw[densely dotted] (-7.2,-12) -- (6.6,-12) (9.4,-12) -- (rd);
    \draw (-7.2,-12) node[left] {$D[d,j \oplus 1]$};
    \draw (14,-14) node[draw, square](rdd) {$d\!+\!1$};
    \draw[thick,-latex] (rd) -- (12,-14);
  \end{scope}
\end{tikzpicture}
\caption{Illustration of definition and computation of the array $D$.}\label{fig:d}
\end{figure}

\begin{fact}
The table $D$ can be computed using the following formula:
 \begin{align*}
 D[0,j] &=\lcp(S, P^\infty[j\dd]),\\
 D[d+1,j] & = i + \lcp(S[i\dd ],P^\infty[i\oplus j\dd]),
 \end{align*}
 where $i = \min(n,\,\max\{D[d,j]+1,\,D[d,j \ominus 1],\,D[d,j \oplus 1]+1\})$.
\end{fact}
\begin{proof}
We will prove the fact by considering the interpretation of $T[i,j]$ as distances in a
weighted graph (see~\cref{sec:wrap}). 
By \cref{lem:Fis}, from every vertex $(i,j)$
we have the following outgoing arcs:
\begin{itemize}
\item $(i,j) \xrightarrow{1} (i + 1, j)$,
\item $(i,j) \xrightarrow{[S[i] \neq P[j]]} (i + 1, j \oplus 1)$,
\item $(i,j) \xrightarrow{1} (i, j \oplus 1)$.
\end{itemize}
Moreover, the value $T[i,j]$ is equal to the minimum distance to $(i,j)$ from
some vertex $(0,j')$.
The only arc of cost 0 is 
$(i,j) \xrightarrow{0} (i+1, j \oplus 1)$
when $S[i] = P[j]$. Therefore, when we have reached a vertex $(i,j)$, the
only vertices we can reach from it by using only 0-arcs are $(i, j),
{(i + 
1, j \oplus 1)}, \ldots, {(i + k, j \oplus k)}$, where $k$ is the
maximum number such that 
$S[i] = P[j]$,
$S[i + 1] = P[j \oplus 1]$,
\ldots,
${S[i + (k-1)] = P[j \oplus (k-1)]}$.
Therefore, $k = \lcp(S[i\dd], P^\infty[j\dd])$.

Hence, $D[0,j] = \lcp(S, P^\infty[j\dd])$ holds for distance $0$.
Taking advantage of monotonicity of distances on each diagonal, we know that full
information about reachable vertices at distance $d$ can be stored as a list of 
the furthest points on each diagonal. Moreover, to reach a vertex of distance $d + 1$
we need to pick a vertex of distance $d$, follow a single 1-arc and then zero or more 0-arcs. 
Combining this with the fact that arcs changing diagonal can be arbitrarily
used at any vertex, it suffices to consider only
the bottom-most point of each diagonal with the distance $d$ as the starting point of the 1-arc, as
we can greedily postpone following an arc that switches diagonals.
\end{proof}

To conclude, assuming we know the indices of furthest reachable cells 
in each of the diagonals for an edit distance $d$, we can easily compute
indices for the next distance. In the beginning, we update the indices by applying
available 1-arcs and afterwards, we increase indices by the results of
appropriate $\lcp$-queries. In the end, we have computed the furthest
reachable cells in each of the diagonals within distance $d + 1$ and achieved
that in linear time with respect to the number of diagonals, i.e., in $\Oh(p)$ time.
This approach is shown as \cref{alg:main}.

\begin{algorithm}[h]
\caption{Compute all values $T[n,j]$ not exceeding $k$}\label{alg:main}
$T[n, 0 \dd p - 1] := (\bot,\ldots,\bot)$\;
\For{$d:=0$ \KwSty{to} $k$}{
    \ForEach{$j\in \mathbb{Z}_p$}{
        \If{$d = 0$}{$i := 0$
        }
        \Else{$i := \min(n,\max(D[d-1, j]+1, D[d-1, j\ominus 1], D[d-1, j\oplus 1]+1))$
        }
        $D[d, j] := i+\lcp(S[i\dd], P^\infty[i\oplus j \dd])$\;
        \If{$D[d,j] = n$ \KwSty{and} $T[n, j \oplus n] = \bot $}{ 
          $T[n, j \oplus n] := d$\;
        }
    }
}
\end{algorithm}

\begin{lemma}
\cref{alg:main} for each $j\in \integ_p$ computes $T[n,j]$ or reports that $T[n,j]>k$.
It can be implemented in $\Oh(n+pk)$ time.
\end{lemma}
\begin{proof}
We use \cref{fct:ver} to answer each $\lcp$-query in constant time,
by creating a data structure for $\lcp$ queries for the word $S\#P^{r }$ where $\#$ is a sentinel character and $r$ is an exponent large enough
so that $|P^r|\ge n+p$.
\end{proof}

\subsection{Main Results}
The table $T$ specifies the last position of an approximate match within the period of the periodic word.
However, in our problem we need to know the starting position, which determines the sought cyclic shift of the period.
Thus, let $T^R$ be the counterpart of $T$ defined for the reverse words $S^R$ and $P^R$. 
Its last row satisfies the following property:
\begin{fact}\label{fct:tr}
For every $j\in \mathbb{Z}_p$, we have 
\[T^R[n,p\ominus j]=\dist(S,U^\infty) \quad\text{where}\quad U=P[j\dd p-1]\cdot P[0\dd j-1].\]
Here, $U$ a cyclic shift of $P$ with the leading $j$ characters moved to the back.
\end{fact}
\begin{proof}
By definition of $T^R$ and $T$, for $0\le i \le n$ and $j\in \integ_p$, we have
\begin{align*}
T^R[n,j] &= \min\{\ed(S^R, (P^R)^\infty[i'\dd j'])\,:\,i' \in \nat,\, j'\equiv j-1\ (\bmod\,p)\}\\
&=\min \{\ed(S, P^\infty[j'\dd i'])\,:\,i' \in \nat,\, j'\equiv -j\ (\bmod\,p)\}\\
& =\min \{\ed(S, P^\infty[p\ominus j\dd i'])\,:\,i' \in \nat\}\\
& =\dist(S, P^\infty[p\ominus j\dd]).
\end{align*}
Consequently, 
\[T^R[n,p\ominus j]=\dist(S,P^\infty[j\dd])=\dist(S, (P[j\dd p-1]\cdot P[0\dd j-1])^\infty)\]
holds as claimed.
\end{proof}

\begin{example}
If $P=\mathtt{ABCA}$ and $S=\mathtt{CBAACAABCA}$ (see~\cref{fig:sroda}), 
then $T^R[10,2] = \dist(\mathtt{CBAACAABCA}, (\mathtt{CAAB})^\infty)=\ed(\mathtt{CBAACAABCA},\mathtt{CAABCAABCA})=2$.
\end{example}

Running \cref{alg:main} for the reverse input, we obtain the solution to the APM problem.

\begin{theorem}\label{thm:main}
  The Approximate Pattern Matching in Periodic Text problem can be solved in $\Oh(n+kp)$ time.
\end{theorem}

By combining \cref{fct:red} and \cref{thm:main} with $k < \tau_p$, we arrive at an improved solution to the APR problem.

\begin{theorem}\label{thm:main2}
  The Approximate Period Recovery problem can be solved in $\Oh(n \log n)$ time.
\end{theorem}

\bibliographystyle{plainurl}
\bibliography{edit_distance}

\end{document}